\documentclass[11pt]{article}
\usepackage{amssymb,amsmath,amsthm}
\newtheorem{theorem}{Theorem}

\newtheorem{proposition}{Proposition}
\newtheorem{lemma}{Lemma}

\newtheorem{cor}{Corollary}

\newtheorem{definition}{Definition}

\newcommand{\prnibble}{{\tt PageRankNibble}}
\newcommand{\prpartition}{{\tt PageRankPartition}}
\newcommand{\justnibble}{{\tt Nibble}}
\newcommand{\nibble}{{\tt RandomNibble}}
\newcommand{\partition}{{\tt Partition}}
\newcommand{\evocut}{{\tt EvoCut}}
\newcommand{\evonibble}{{\tt EvoNibble}}
\newcommand{\evopartition}{{\tt EvoPartition}}
\newcommand{\fastsample}{{\tt GenerateSample}}

\newcommand{\sumphi}[1]{Q_{#1}}
\newcommand{\good}[2]{{#1_{#2}}}

\newcommand{\polylog}[1]{ {(#1) \cdot O(\operatorname{polylog}(n))}}
\newcommand{\polylognoparen}[1]{ {O(#1 \operatorname{polylog}(n))}}

\def \timecost{ \operatorname{cost}}

\def \targetset{A}
\def \Jmax{J_{max}}

\newcommand{\comp}[1]{{ {#1}^{c} }}

\newcommand{\setdiff}[2]{{S_{#1} \Delta S_{#2}}}

\newcommand{\ind}{\mathbf{1}}

\def\P {{ \bf P}}
\def\Pwalk {{\mathcal{P}}}
\def\E {{ \bf {E}}}

\def\newS { S_{1}}

\def\vol {\mu}
\newcommand{\setvol}[1]{ { \mu(S_{#1}) } }

\newcommand{\C}[1]{{\timecost(\ess{#1})}}

\newcommand{\R}[1]{{R_{#1}}}

\def \half{{\textstyle{1\over2}}}

\def\K {{ \bf K}}
\def\Kh {{ \bf {\widehat K}}}

\def\P     {{\bf P}}
\def\Ph    {{\bf \widehat P}}
\def\Eh    {{\bf \widehat E}}
\def\Pcoup {{\bf P}^{*}}
\def\Kcoup {{\bf K}^{*}}

\def\e {{ \bf {E}}}

\newcommand{\ess}[1]{{ S_{0}, \ldots, S_{#1} }} 
\newcommand{\path}[2]{{ S_{#1}, \ldots, S_{#2} }}

\newcommand{\exit}[3]{\operatorname{esc}(#1,#2,#3)}

\newcommand{\ignore}[1]{}

\begin{document} 

\begin{titlepage}
\title{Finding Sparse Cuts Locally Using Evolving Sets}
\author{Reid Andersen and Yuval Peres}
%\institute{Microsoft Research}
%\email{reidan@microsoft.com}  
%\email{peres@microsoft.com}  

\maketitle 

\begin{abstract}

A {\em local graph partitioning algorithm} finds a set of vertices with small
conductance (i.e. a sparse cut) by adaptively exploring part of a large graph $G$,
starting from a specified vertex.  For the algorithm to be local, its complexity must be bounded
in terms of the size of the set that it outputs, with at most a weak dependence on the number $n$ of vertices in $G$.
Previous local partitioning algorithms find sparse cuts using random walks and personalized PageRank.
In this paper, we introduce a randomized local partitioning
algorithm that finds a sparse cut by simulating the {\em volume-biased evolving
set process}, which is a Markov chain on sets of vertices.  We
prove that for any set of vertices $A$ that has conductance at most $\phi$, for at least half 
of the starting vertices in $A$ our algorithm will output (with probability at least half), 
a set of conductance $O(\phi^{1/2} \log^{1/2} n)$. 
We prove that for a given run of the algorithm, 
the expected ratio between its computational complexity 
and the volume of the set that it outputs is $\polylognoparen{\phi^{-1/2}}$.  
In comparison, the best previous local
partitioning algorithm, due to Andersen, Chung, and Lang, has the same approximation guarantee, but a larger ratio
of $\polylognoparen{\phi^{-1}}$ between the complexity and output volume.  Using our
local partitioning algorithm as a subroutine, we construct a fast algorithm
for finding balanced cuts.  Given a fixed value of $\phi$, the resulting
algorithm has complexity $\polylog{m+n\phi^{-1/2})}$ and returns a cut with
conductance $O(\phi^{1/2} \log^{1/2} n)$ and volume at least $v_{\phi}/2$, where $v_{\phi}$ is
the largest volume of any set with conductance at most $\phi$.
\end{abstract}

\end{titlepage}

\section{Introduction}

A {\em local graph partitioning algorithm} solves a targeted version of the classic sparsest cut problem; 
it finds a set with small conductance by adaptively examining a small subset of the input graph near a specified starting vertex. 
Such algorithms are useful for finding target clusters in large graphs, and for quickly finding collections of small clusters.
They have been applied in practice to probe the community structure of social and information networks~\cite{Flake:2000, AndersenLang:2006, Leskovec:2008},
and have been used as subroutines to design fast algorithms for other partitioning problems~\cite{Spielman:2004,Spielman:2008}.  

Spielman and Teng introduced a local
partitioning algorithm with a remarkable approximation guarantee and bound on its computational complexity~\cite{Spielman:2004,Spielman:2008}. 
Their algorithm has a bounded {\em work/volume ratio}, which is the
ratio between the work performed by the algorithm on a given run (meaning the number of operations or computational complexity),
and the volume of the set it outputs.  It also has a {\em local approximation guarantee}, which states
(roughly) that if the starting vertex is contained in a set with conductance at
most $\phi$, then the algorithm will output a set with conductance at most
$f(\phi)$.  To find such a set, their algorithm computes a sequence of vectors that
approximate the sequence of probability distributions of a random walk from the
starting vertex.  The support of these vectors is kept small by removing tiny
amounts of probability mass at each step.  The most recent version of their
algorithm~\cite{Spielman:2008} has local approximation guarantee $f(\phi) =
O(\phi^{1/2} \log^{3/2} n)$ and work/volume ratio $\polylognoparen{\phi^{-2}}$.  Andersen,
Chung, and Lang~\cite{Andersen:2006} introduced a local partitioning algorithm
that computes a single personalized PageRank vector rather than a sequence of
random walk distributions.  Their algorithm has approximation guarantee
$O(\phi^{1/2} \log^{1/2} n)$ and work/volume ratio $\polylognoparen{\phi^{-1}}$.  

The {\em evolving set process} (ESP) is a Markov chain whose states are subsets of the vertex set of a graph. Its transition rule is a simple procedure that grows or shrinks the current set.  Morris and Peres used the ESP, and the closely related {\em volume-biased evolving set process} (volume-biased ESP), to bound the mixing time of Markov chains in terms of their isoperimetric properties~\cite{Morris:2003}.  The volume-biased ESP is equivalent to the {\em strong stationary dual} of a random walk, which was introduced earlier by Diaconis and Fill~\cite{Diaconis:1990}.  Further applications of evolving sets were described in~\cite{Montenegro, Montenegro2}.  In all of these results, evolving sets were used as analytical tools rather than algorithms.

In this paper, we design a local partitioning algorithm called $\evocut$ based on evolving sets.  Our algorithm simulates the volume-biased evolving set process until a certain stopping time is reached, then outputs the resulting set.  We prove that the algorithm has local approximation guarantee 
$O(\phi^{1/2} \log^{1/2} n)$ and expected work/volume ratio $\polylognoparen{\phi^{-1/2}}$. To prove the local approximation guarantee, we bound the rate of growth of the sets in the volume-biased ESP.  In particular, we prove a lower bound that depends on the conductance of the sets observed by the process, and an upper bound that depends on the conductance of certain sets that contain the starting vertex.  To bound the work/volume ratio, we combine a simple implementation trick with a nontrivial probabilistic analysis.  We introduce an efficient method for simulating the volume-biased ESP that updates the vertices on the boundary of the current set and ignores the vertices in the interior.  The work required to generate a sample path using this method is proportional to the {\em cost} of the sample path, which depends on the boundaries of the sets observed and the symmetric differences between successive sets.  Using a martingale argument, we prove that the expected ratio between the cost of a sample path and the volume of the set output is $\polylognoparen{\phi^{-1/2}}$, which bounds the work/volume ratio of our algorithm.  The main theorem about $\evocut$, which gives a precise statement of its work/volume ratio and local approximation guarantee, 
is stated in Section~\ref{mainresult}.  In Table~\ref{fig:local}, we compare $\evocut$ with existing local partitioning algorithms.  
\begin{table}
  \begin{center}
\begin{tabular}{|l|ll|}
  \hline
  local partitioning algorithm &  work/volume ratio &  approximation guarantee\\
    \hline
    $\justnibble$ (ST04)~\cite{Spielman:2004}        & $\polylognoparen{\phi^{-5/3}}$  & $\phi \rightarrow O(\phi^{1/3}\log^{2/3}n)$\\
    $\justnibble$ (ST08) \cite{Spielman:2008}               & $\polylognoparen{\phi^{-2}}$      & $\phi \rightarrow O(\phi^{1/2}\log^{3/2}n)$\\
    $\prnibble$  (ACL06) \cite{Andersen:2006}             & $\polylognoparen{\phi^{-1}}$      & $\phi \rightarrow O(\phi^{1/2}\log^{1/2}n)$\\
    $\evocut$  (this paper)                      & $\polylognoparen{\phi^{-1/2}}$    & $\phi \rightarrow O(\phi^{1/2}\log^{1/2}n)$\\
    \hline
\end{tabular}\hfill
\caption{The work/volume ratio and approximation guarantee of known local partitioning algorithms.  Here $n = |V|$ is the number of vertices in the graph.}\label{fig:local}
\end{center}
\end{table}

One application of our local partitioning algorithm is a fast algorithm for
finding balanced cuts.  Spielman and Teng showed how to find a balanced cut in
nearly linear time by repeatedly removing small sets from a graph using local
partitioning~\cite{Spielman:2004}.  Applying their technique with our algorithm
yields an algorithm $\evopartition$ with the following properties.  The
algorithm has complexity $\polylog{m + n\phi^{-1/2}}$, and it outputs a set of
vertices whose conductance is $O(\phi^{1/2} \log^{1/2} n)$ and whose volume at
least half that of any set with conductance at most $\phi$, where $\phi$ is an
input to the algorithm.  Our algorithm is faster by a factor of roughly
$\phi^{1/2}$ than any existing algorithm that provides a nontrivial
approximation guarantee for the balanced cut problem, but there are several
algorithms that provide stronger approximation guarantees.  The fastest
previously known algorithms for finding balanced cuts are due to
Arora-Kale~\cite{Arora:2007} and Orecchia et al.~\cite{Orecchia:2008}.  These
algorithms produce cuts with conductance $O(\phi \log n)$, and their
computational complexity is dominated by the cost of solving
polylogarithmically many single-commodity flow problems, namely $\polylog{m +
\min(n/\phi,n^{3/2})}$.  In Section~\ref{sec:balanced},
we give a more detailed description of $\evopartition$ and 
comparison with existing balanced cut algorithms.

In the remainder of this section, we state the main theorem about our local partitioning algorithm $\evocut$.
In section~\ref{sec:prelim}, we review the basic properties of the ESP and volume-biased ESP.
In section~\ref{sec:partition}, we show how to find cuts with small conductance by generating sample paths from the volume-biased ESP.
In section~\ref{sec:sample}, we describe an algorithm for simulating the volume-biased ESP.
We then construct $\evocut$ and prove the main theorem about its work/volume ratio and local approximation guarantee.
In section~\ref{sec:balanced}, we describe the balanced cut algorithm $\evopartition$.

\subsection{Main result}\label{mainresult}
Let $G=(V,E)$ be a simple undirected graph with $n=|V|$ vertices and $m=|E|$ edges.
The {\bf volume} $\vol(S)$ of a set of vertices $S \subseteq V$ 
is defined to be \[\mu(S) := \sum_{x \in S}d(x),\]
where $d(x)$ denotes the degree of the vertex $x$.
The number of edges between two sets of vertices $S$ and $R$ is written $e(S,R)$.
The complement of $S$ is written $S^{c} = V \setminus S$,
and we define $\partial(S) = e(S,\comp{S})$ to be the number of edges leaving $S$.
The {\bf conductance} of a set of vertices $S$ is defined to be
\[\phi(S) := \partial(S)/\vol(S).\]  
Notice that $\phi(V) = 0$. 
In other papers, the conductance of a set is sometimes defined to be $\partial(S)/\min(\vol(S),\vol(\comp{S}))$.
When a set is output by one of our partitioning algorithms, 
we will upper bound its volume by $(3/4)\mu(V)$, which ensures that the two definitions of conductance differ by only a constant factor.
When the base is omitted, $\log$ means $\log_{e}$.

Our main result is the analysis of the local partitioning algorithm $\evocut$.
The algorithm makes queries to an input graph $G=(V,E)$.  We assume the graph supports the following types of queries, which would be easy to support in practice by storing the graph in random access memory.  Given an arbitrary vertex $x$, let $N(x)$ be the set of vertices adjacent to a given vertex $x$.  We assume we can obtain a list of the vertices in $N(x)$ in time proportional to $|N(x)|$, and obtain a node sampled uniformly from $N(x)$ in constant time.
The following is the main theorem, which describes the work/volume ratio and local approximation guarantee of $\evocut$.  
\begin{theorem}\label{thm:evocut}
    $\evocut(v,\phi)$ takes as input a {\em starting vertex} $v \in V$ and a 
    {\em target conductance} $\phi \in (0,1)$, and outputs a set of vertices.
    For a given run of the algorithm, let $S$ be the set of vertices it outputs, and let $w$ be the amount of work it performs (the computational complexity).
   Both $S$ and $w$ depend on randomness used by the algorithm.
  \begin{enumerate}
      \item  
          Let $w/\vol(S)$ be the work/volume ratio.
          Then, 
          \[\E[w/\vol(S)] = O(\phi^{-1/2} \log^{3/2} |V|).\]

      \item 
          If $\targetset \subseteq V$ is a set of vertices that satisfies 
          \mbox{$\phi(\targetset) \leq \phi$}
          and 
          $\vol(\targetset) \leq (2/3) \vol(V)$,
          then there is a subset 
          $\targetset' \subseteq \targetset$ with volume at least $\vol(\targetset)/2$
          such that whenever $v \in \targetset'$, 
          with probability at least $1/2$
          the output set $S$ satisfies all of the following:
          \begin{enumerate}
            \item $\phi(S) = O(\phi^{1/2} \log^{1/2}|V|)$,
            \item $\vol(S) \leq (3/4)\vol(V)$,
            \item $\vol(S \cap \targetset) \geq (9/10) \vol(S)$.
          \end{enumerate}
  \end{enumerate}
\end{theorem}
The description of $\evocut$ and the proof of Theorem~\ref{thm:evocut} are given in Section~\ref{sec:sample}. 

\section{Preliminaries}\label{sec:prelim}

In this section we describe the ESP and volume-biased ESP,
the connections between them, and their relationship to conductance and random walks.  We use the terminology and basic results from~\cite{Morris:2003}.  The coupling described in section~\ref{sec:coupling} is due to Diaconis-Fill~\cite{Diaconis:1990}.  The volume-biased ESP is equivalent to one of the strong stationary duals constructed in~\cite{Diaconis:1990}, which predates the ESP and volume-biased ESP.

\subsection{Random Walk}
A random walk on the graph $G$ is a Markov chain
defined by the transition kernel
\[p(x,y) 
= 
\begin{cases}
    1/(2d(x)) & \text{if $\{x,y\} \in E$},\\
    1/2 	& \text{if $x = y$},\\
	0  	& \text{otherwise}.\\
\end{cases}
\]
Note that this is a ``lazy'' walk with holding probability $1/2$.
Given a set $S$, we let $p(x,S)$ denote the
probability of transitioning from $x$ to some vertex in $S$,
\[ p(x,S) := \sum_{y \in S}p(x,y) = 
\frac{1}{2}\left(\frac{e(x,S)}{d(x)} + \ind(x \in S) \right).
\]
Here, $\ind(\cdot)$ denotes the indicator function for an event.
We write $p^{t}(x,y)$ for the $t$-step transition probabilities,
and let $\Pwalk_{x}$ denote the probability measure for the Markov chain of a random walk started from $x$.

\subsection{The Evolving Set Process} 

The {\bf evolving set process} (ESP) is a Markov chain on subsets of the vertex set $V$.
Given the current state $S$, the next state $\newS$ is chosen by the following rule:
pick a threshold $U$ uniformly at random from the interval $[0,1]$,
and let 
\begin{align}\label{E:transition}
    \newS &= \{ y: p(y,S) \geq U \}.
\end{align}
Notice that $\emptyset$ and $V$ are absorbing states for the process.
Given a starting state $S_{0} \subseteq V$, we write $\P_{S_{0}}( \cdot ):= \P(\cdot \mid S_0 )$ to denote the probability measure for the ESP Markov chain started from $S_{0}$. Similarly, we write $\E_{S_{0}}( \cdot )$ for the expectation.  
For a singleton set, we use the shorthand $\P_{x}( \cdot ) = \P_{\{x\}}( \cdot ) $.
We define the transition kernel $\K(S,S') = \P_{S}(\newS = S')$.

\subsection{Evolving sets and conductance}\label{sec:conductance}

The following propositions relate the conductance of a set in the ESP to the change in volume in the next step.
The first proposition strengthens the fact that the sequence $(\vol(S_t))_{t \geq 0}$ is a martingale.
   
\begin{proposition}\label{prop:conditional}
Let $U$ be the uniform random variable used to generate
$\newS$ from $S$ in the ESP. 
Then,
\[
\e_S( \vol(\newS) \mid  U \leq \half) = \vol(S) + \partial(S) = \vol(S)(1 + \phi(S)).
\]
\[
\e_S( \vol(\newS) \mid  U > \half) = \vol(S) - \partial(S) = \vol(S)(1 - \phi(S)).
\]
\end{proposition} 
\begin{proposition}\label{prop:gauge}
The {\bf growth gauge} $\psi(S)$ of a set $S$ is defined by the following equation:
\[1-\psi(S) := \E_{S} \sqrt{\frac{\vol(\newS)}{\vol(S)}}. \]
For any set $S \subseteq V$, the growth gauge and conductance satisfy $\psi(S) \geq \phi(S)^{2}/8$.
\end{proposition} 
Proofs of Propositions~\ref{prop:conditional} and~\ref{prop:gauge} appear in~\cite{Morris:2003}, but the constants stated there differ from ours; their definition of conductance incorporates the holding probability from the random walk, which makes it smaller than ours by a factor of 2.

\subsection{The Volume-Biased Evolving Set Process}
The {\bf volume-biased evolving set process} (volume-biased ESP) is a Markov chain on subsets of $V$ with the following transition kernel:
\begin{equation}\label{eqn:hat}
    \Kh(S,S') = \frac{\vol(S')}{\vol(S)} \K(S,S'),
\end{equation}
where $\K(S,S')$ is the 
transition kernel for the ESP.
We remark that $\Kh$ is the {\em Doob $h$-transform} of $\K$ with respect to $\vol$ (see chapter 17 of~\cite{MarkovChainsAndMixingTimes}),
and that the volume-biased ESP is equivalent to the ESP conditioned to absorb in the state $V$.
Given a starting state $S_{0}$, 
we write $\Ph_{S_0}( \cdot ):= \Ph(\cdot \mid S_0 )$ 
for the probability measure of the Markov chain.
Similarly, we write $\Eh_{S_{0}}( \cdot )$ for the expectation.  

The following proposition relates the volume-biased ESP and the ESP.
This is a standard consequence of the Doob $h$-transform, but we include a proof for completeness.
\begin{proposition}\label{prop:transform}
For any function $f$ and any starting set $S_{0} \neq \emptyset$,
\begin{align}
\label{E:hatexpectation}
\Eh_{S_0} \left[ f(\ess{n}) \right] = \e_{S_0} \left[ \frac{\vol(S_n)}{\vol(S_0)} f(\ess{n}) \right].
\end{align}
\end{proposition}
\begin{proof}
Assume that $S_{0} \neq \emptyset$.
Let $\mathcal{C}$ be the collection of sample paths $(\path{0}{t})$ such that $\Ph_{S_{0}}(\path{1}{t}) > 0$.
If $(\path{0}{t}) \in \mathcal{C}$, then $\mu(S_{j})>0$ for all $j \in [0,t]$, so  
    \begin{align*}
        \Ph_{S_{0}}(\path{1}{t}) 
        = \prod_{j=0}^{t-1} \frac{\vol(S_{j+1})}{\vol(S_{j})} 
        \P_{S_{j}}(S_{j+1})
        = \frac{\vol(S_t)}{\vol(S_0)} \P_{S_0}(\path{1}{t}).  
    \end{align*}
Therefore,
\begin{align*}\Eh_{S_{0}} \left[ f(\path{0}{t}) \right] 
&= 
\sum_{(\path{0}{t}) \in \mathcal{C}} f(\path{0}{t}) \Ph_{S_{0}}(\path{1}{t})\\
&=
\sum_{(\path{0}{t}) \in \mathcal{C}} f(\path{0}{t}) \frac{\vol(S_t)}{\vol(S_0)} \P_{S_{0}}(\path{1}{t})\\
&=
\e_{S_0} \left[ \frac{\vol(S_t)}{\vol(S_0)} f(\path{0}{t}) \right].
\end{align*}
\end{proof}

\subsection{The Diaconis-Fill Coupling}\label{sec:coupling}
Diaconis-Fill~\cite{Diaconis:1990} introduced the following coupling between the random walk process and the volume-biased ESP.
Let $(X_{t},S_{t})$ be a Markov chain, where $X_{t}$ is a vertex and $S_{t} \subseteq V$ is a subset of vertices.
Let $\Pcoup$ be the probability measure for the Markov chain.
Given a starting vertex $x$, let $X_{0} = x$ and $S_{0} = \{x\}$, and
let $\Pcoup_{x}(\cdot) = \Pcoup(\cdot \mid X_{0} = x, S_{0} = \{x\})$.
Given the current state $(X_{t},S_{t})$, the transition probabilities are defined as follows.
\begin{align*}
  \Pcoup(X_{t+1} = y' \mid X_{t} = y, S_{t} = S) &= p(y,y'),\\
  \Pcoup(S_{t+1} = S' \mid S_{t} = S, X_{t+1} = y') &= 
    \frac{\K(S,S') \ind(y' \in S')}{\P(y' \in S_{t+1} \mid S_{t} = S)}. 
  \end{align*}
In words, we first select $X_{t+1}$ according to the random walk transition kernel,
then select $S_{t+1}$ according to the ESP transition kernel restricted to sets that contain $X_{t+1}$.
We define the transition kernel $\Kcoup((y,S),(y',S')) = \Pcoup(X_{1}=y', S_{1} = S' \mid X_{0} = y, S_{0} = S )$.

The following proposition shows that $\Pcoup_{x}$ 
is a coupling between the random walk process and the
volume-biased ESP, and furthermore the
distribution of $X_{t}$ conditioned on $(\path{0}{t})$ is the stationary distribution restricted to $S_{t}$.  
A proof of Proposition~\ref{prop:coupling} is given in chapter 17 of~\cite{MarkovChainsAndMixingTimes}.
\begin{proposition}[Diaconis and Fill]\label{prop:coupling}
  Let $(X_{t},S_{t})$ be a Markov chain started from $(x,\{x\})$ with the transition kernel $\Kcoup$.
\begin{enumerate}
\item
    The sequence $(X_{t})$ is a Markov chain started from $x$ with the transition kernel $p(\cdot,\cdot)$.
\item
    The sequence $(S_{t})$ is a Markov chain started from $\{x\}$ with transition kernel $\Kh$.
\item For any vertex $y$ and time $t \geq 0$,
    \[\Pcoup_{x}( X_{t} = y \mid \path{1}{n} ) = \ind(y \in S_{t})\frac{d(y)}{\vol(S_{t})}.\]
\end{enumerate}
\end{proposition}

\section{Local partitioning using the volume-biased evolving set process}\label{sec:partition} 
In this section, we show how to find sets with small conductance by generating sample paths from the volume-biased ESP.
The following theorem shows that if we start from a single starting vertex and simulate the volume-biased ESP for $T$ steps,  then one of the states observed is likely to have conductance $O(\sqrt{T^{-1} \log n})$.
We can also prove that all the states observed are likely to have volume at most $(3/4)\mu(V)$,
provided there exists a set $\targetset \subseteq V$ that has conductance at most $T^{-1}$, and that the starting vertex belongs to a certain subset of $A$.
\begin{theorem}\label{thm:partition}
  Fix an integer $T$, and let $\targetset \subseteq V$ be any set of vertices
  that satisfies $\vol(\targetset) \leq (2/3) \vol(V)$ and $\phi(\targetset) \leq (100T)^{-1}$.
  Then, there exists a subset $\good{\targetset}{T} \subseteq \targetset$ of volume at least $\vol(\targetset)/2$ for which the following holds.
  If $x \in \targetset_{T}$, then with probability at least $7/9$ a sample path $(\path{1}{n})$ 
  from the volume-biased ESP started from $S_{0} = \{x\}$
  will satisfy all of the following:
  \begin{enumerate}
    \item $\phi(S_{t}) < 3\theta_{T}$ for some $t \in [0,T]$, where $\theta_{T} = \sqrt{4T^{-1}\log \mu(V)}$.
      \item $\vol(S_{j}) \leq (3/4) \vol(V)$ for all $j \in [0,T]$.
      \item $\vol(S_{j} \cap \targetset) \geq (9/10) \vol(S_{j})$ for all $j \in [0,T]$.
  \end{enumerate}
\end{theorem}
The proof of Theorem~\ref{thm:partition} is at the end of this section, after we present two necessary Lemmas.

Consider a sample path from the volume-biased ESP.  The following lemma shows it is unlikely for the sample path to contain many sets with large conductance.  
Intuitively, this is true because at each step the quantity $\vol(S_{t})$ tends to increase at a rate that depends on $\phi(S_{t})$.
Eventually the sample path will absorb in the state $V$, whose conductance is $\phi(V) = 0$.
\begin{lemma}\label{L:expectation}
For any starting set $S_{0} \subseteq V$ and any stopping time $\tau$ for the volume-biased ESP,
\[
\Eh_{S_{0}} \left[ \sum_{j=0}^{\tau} \phi(S_{j})^{2} \right]
\leq 
4 \Eh \log \frac{{\vol(S_{\tau})}}{\vol(S_{0})}
\leq 
4 \log \vol(V).
\]
\end{lemma}

\begin{proof}
Recall from Proposition~\ref{prop:gauge} that, by definition, $1-\psi(S) := \E_{S} \sqrt{\vol(S_1)/\vol(S)}$. Then,
\begin{align*}
  \Eh_{S_{t-1}} \left (\frac{1}{\sqrt{\setvol{t}}} \Big| \, S_{t-1} \right )
    &= \E_{S_{t-1}} \left (\frac{\setvol{t}}{\setvol{t-1}} \frac{1}{\sqrt{\setvol{t}}} \Big| \, S_{t-1}  \right ) \\
    &= \frac{1}{\setvol{t-1}} \E_{S_{t-1}} \left (\sqrt{\setvol{t}} \Big| \,  S_{t-1} \right ) \\
    &= \frac{1}{\setvol{t-1}} (1 - \psi(S_{t-1})) \sqrt{\setvol{t-1}}
     = \frac{1 - \psi(S_{t-1})}{\sqrt{\setvol{t-1}}}.
\end{align*}
We define
\begin{equation}\label{Mgrowth}
M_{t} := F_{t} \frac{\sqrt{\setvol{0}}}{\sqrt{\setvol{t}}},
\quad \text{where } F_{t} := \prod_{j=0}^{t-1} (1-\psi(S_{j}))^{-1},
\text{ and } F_{0} := 1,
\end{equation}

We now verify that $(M_{t})$ is a martingale in the volume-biased ESP:
\begin{align*}
  \Eh \left ( M_{t} \Big| \ess{t-1} \right ) 
  &=
  F_{t} 
  \Eh 
  \Bigl (
    \frac{\sqrt{\setvol{0}}}{\sqrt{\setvol{t}}} \, \Big| \,  S_{t-1} 
  \Bigr ) \\
  &=
  F_{t} (1-\psi(S_{t-1})) 
  \frac{\sqrt{\setvol{0}}}{\sqrt{\setvol{t-1}}}\\
  &= 
  F_{t-1} \frac{\sqrt{\setvol{0}}}{\sqrt{\setvol{t-1}}} 
  = 
  M_{t-1}.
\end{align*}
Let $\tau$ be a stopping time for the volume-biased ESP.
By the optional stopping theorem for nonnegative martingales (see~\cite{Williams}), we have $\Eh M_{\tau} \leq M_{0} = 1$.
Then by Jensen's inequality, we have $\Eh \log M_{\tau} \leq \log(\Eh M_{\tau}) = 0$.
Taking the log of \eqref{Mgrowth},
\begin{equation}\label{Mlog}
\log F_{\tau}
= 
\log M_{\tau} + \frac{1}{2} \log \frac{\vol(S_{\tau})}{\vol(S_{0})}. 
\end{equation}

Since $(1-\psi(S_{j}))^{-1} \geq e^{\psi(S_{j})}$,
\begin{equation}\label{logF}
\log F_{\tau} 
= \log \prod_{j=0}^{\tau-1} (1-\psi(S_{j}))^{-1}
\geq \sum_{j=0}^{\tau-1} \psi(S_{j}).
\end{equation}

Taking expectations in \eqref{Mlog} and \eqref{logF} yields 
\[
\Eh \sum_{j=0}^{\tau-1} \psi(S_{j})
\leq
\Eh \log F_{\tau}
=
\Eh \log M_{\tau} + \frac{1}{2} \Eh \log \frac{\vol(S_{\tau})}{\vol(S_{0})}
\leq 
\frac{1}{2} \Eh \log \frac{{\vol(S_{\tau})}}{\vol(S_{0})}.
\]

We apply the inequality $\phi(S_{j})^{2} \leq 8\psi(S_{j})$ from
Proposition~\ref{prop:gauge} to finish the proof, 
\[
\Eh \sum_{j=0}^{\tau-1} \phi(S_{j})^{2}
\leq 
8\Eh \sum_{j=0}^{\tau-1} \psi(S_{j})
\leq 
4 \Eh \log \frac{\vol(S_{\tau})}{\vol(S_{0})}
\leq 
4\log \vol(V).
\]
\end{proof}

\begin{cor}\label{L:minconductance}
Let $\theta_{T} = \sqrt{4T^{-1}\log \mu(V)}$.
For any starting set $S_{0}$, integer $T$, and constant $c \geq 0$,
\[
\Ph_{S_{0}} \left[ \min_{j < T} \phi(S_{j}) \leq \sqrt{c} \theta_{T} \right]
\geq
1-1/c.
\]
\end{cor}
\begin{proof}
Fix $S_{0}$ and $T$, and consider a sample path $(\path{0}{T})$ from the volume-biased ESP.
Let $\phi_{j} := \phi(S_{j})$.
Lemma~\ref{L:expectation} implies that
$\Eh_{S_{0}} \left[ \sum_{j<T}{\phi^{2}_{j}} \right] \leq 4 \log \vol(V)$.
By Markov's inequality, 
the event $\sum_{j<T}{\phi^{2}_{j}}\leq 4c\log \vol(V)$ holds with probability at least $1-1/c$.  
If that event holds, then $\min_{j<T}(\phi_{j})  \leq \sqrt{c} \theta_{T}$.
\end{proof}
  
Now that we know a sample path $(\path{0}{T})$ from the volume-biased ESP is likely to contain a set with small conductance, 
we are halfway done with the proof of Theorem~\ref{thm:partition}.  We still need to 
show that the sets observed in the volume-biased ESP are likely to have volume at most $(3/4)\vol(V)$.
We start with a standard fact (Proposition~\ref{prop:exitprob}) that 
bounds the probability that a lazy random walk escapes from a given set $\targetset \subseteq V$.
We then prove Lemma~\ref{L:containment}, which converts this standard fact into a statement about the volume-biased ESP.
By combining these results we obtain a bound on the fraction of $S_{j}$ that is not contained in $\targetset$. 
This yields a bound on the total volume of $S_{j}$.
\begin{proposition}\label{prop:exitprob}
  Let $(X_{i})$ be a lazy random walk Markov chain starting from the vertex $x$.
  For any set $\targetset \subseteq V$ and integer $T$, let
  \[
  \exit{x}{T}{\targetset}
  :=
  \Pwalk_{x} \left[ \cup_{j=0}^{T}(X_{j} \not \in \targetset) \right],
  \]
  which is the probability that a lazy random walk starting from $x$
  leaves $\targetset$ within the first $T$ steps, and define
  \mbox{$\good{\targetset}{T} := \{ x \in A \mid  \exit{x}{T}{\targetset} \leq T\phi(\targetset) \}$.}
  Then, $\vol( \good{\targetset}{T} ) \geq (1/2)\vol(\targetset)$.
\end{proposition}
A proof of Proposition~\ref{prop:exitprob} appears in~\cite{Spielman:2008}. 
In Theorem 2.5 of that paper, it is shown that 
$\vol(\targetset)^{-1} \sum_{x \in \targetset} \vol(x) \exit{x}{T}{\targetset} \leq T\phi(\targetset)/2$.
This implies $\vol(A_{T}^{c}) \leq \vol(A)/2$, by Markov's inequality.
The statement of Theorem 2.5 is slightly weaker due to a minor difference in the definition of $\phi$, 
but their proof establishes the stronger statement.

\begin{lemma}\label{L:containment}
  For any vertex $x$, set $\targetset \subseteq V$, and integer $T\geq0$,
  the following holds for all $\lambda > 0$,
  \[
  \Ph_x \left[ 
  \max_{t \leq T} 
  \frac{ \vol( S_{t} \setminus \targetset)}{\vol (S_{t})}
  > \lambda \exit{x}{T}{\targetset}
  \right ]
  < \frac{1}{\lambda}.
  \]
\end{lemma}

\begin{proof}
  Recall from section~\ref{sec:coupling} the coupling $\Pcoup$ between the volume-biased ESP Markov chain $(S_{t})$ and the random walk Markov chain $(X_{t})$. 
  This coupling has the property that for any $t \geq 0$,
  \[
  \Pcoup \left[ X_{t} = y \mid \path{0}{t} \right] 
  = \frac{d(y)}{\vol(S_{t})} \ind(y \in S_{t}).
  \]
  Fix a value $\gamma \in [0,1]$ and let $\tau$ be the first time $t$ when
  $\vol( S_{t} \setminus \targetset) > \gamma \vol (S_{t})$, or let $\tau = \infty$ if this does not occur.
  Consider the probability that $X_{\tau} \not \in \targetset$,
  conditioned on $S_{\tau}$:
  \begin{align*}
    \Pcoup \left[ X_{\tau} \not \in \targetset \mid S_{\tau} = S \right] 
    = \sum_{y \in S \setminus \targetset} \frac{d(y)}{\vol(S)}
    = \frac{ \vol( S\setminus \targetset)}{\vol (S)}.
  \end{align*}
  By the definition of $\tau$, we have $\Pcoup \left[ X_{\tau} \not \in \targetset \mid \tau \leq T \right] > \gamma$, so
  \begin{align*}
    \exit{x}{T}{\targetset}
    &=
    \Pcoup \left[ \cup_{j=0}^{T}(X_{j} \not \in \targetset) \right]\\
    &\geq
    \Pcoup \left[ X_{\tau} \not \in \targetset \right ]\\
    &\geq
    \Pcoup \left[ X_{\tau} \not \in \targetset \mid \tau \leq T \right]
    \Pcoup [ \tau \leq T ]\\
    &>
    \gamma \Pcoup [ \tau \leq T ].\
  \end{align*}
  Therefore
  \[
  \Ph_x \left[ 
  \max_{t \leq T} 
  \frac{ \vol( S_{t} \setminus \targetset)}{\vol (S_{t})}
  > \gamma
  \right ]
  =
  \Pcoup [ \tau \leq T ]
  <
  \frac{\exit{x}{T}{\targetset}}
  {\gamma}.
  \]
  The lemma follows by taking $\gamma = \lambda \exit{x}{T}{\targetset}$.
\end{proof}

We now combine the results from this section to prove Theorem~\ref{thm:partition}.
\begin{proof}[Proof of Theorem~\ref{thm:partition}]
  Let $\targetset \subseteq V$ be a set and $T$ be an integer that satisfy $\phi(\targetset)\leq (100T)^{-1}$. Let $\targetset_{T} \subseteq \targetset$ be the set defined in Proposition~\ref{prop:exitprob},
and assume that $x \in \good{\targetset}{T}$. Let $(\path{0}{T})$ be a sample path from the volume-biased ESP started from $S_{0}=\{x\}$.

By Corollary~\ref{L:minconductance}, with probability at least $1-1/9$ there exists some $t<T$ for which
$\phi(S_{t}) \leq 3\theta_{T}$. 
The definition of $A_{T}$ implies that $\exit{x}{T}{\targetset} \leq T\phi(\targetset) \leq 1/100.$
Lemma~\ref{L:containment} then shows that with probability at least $9/10$,
  \begin{equation*} 
      \frac{ \vol( S_{j} \setminus \targetset)}{\vol (S_{j})} \leq 
      10 \exit{x}{T}{\targetset} \leq \frac{1}{10} \quad \text{for all $j \in [0,T]$}.
  \end{equation*}
  Since $\vol(\targetset) \leq (2/3) \vol(V)$, we have, for all $j \in [0,T]$,
  \[ 
  \vol(S_{j}) 
  \leq 
  \frac{\vol(S_{j})}{\vol(S_{j} \cap \targetset)} \vol(\targetset)
  \leq (10/9)(2/3) \vol(V)
  \leq (3/4)\vol(V).
  \]
  By the union bound, with probability at least $7/9$
  the sample path $(\ess{T})$ satisfies all the conclusions of the theorem.
\end{proof}

\section{Simulating the volume-biased evolving set process}\label{sec:sample}

In the beginning of this section, we describe a subroutine $\fastsample$ that 
simulates the volume-biased ESP until a certain stopping time $\tau$ is reached, generating a sample path $(\ess{\tau})$ and producing as output the set $S_{\tau}$.  
We choose $\tau$ to be the first time that $S_{t}$ has sufficiently small conductance, or that the work performed exceeds a specified limit.  
We assign a cost to each sample path that depends on the boundaries of the sets in the path, and the difference in volume between successive sets.
We then show that the work performed by $\fastsample$ is $O(\operatorname{polylog}(n))$ times the cost of the sample path it generates.  

The algorithm $\evocut$, which we construct at the end of this section, outputs the set $S_{\tau}$ computed by $\fastsample$.
To bound the work/volume ratio of $\evocut$, we directly bound the expected ratio between the cost of $(\ess{\tau})$ and the volume of $S_{\tau}$.  

\begin{definition}
  The {\bf cost} of a sample path $(\ess{t})$ is 
  \begin{align} \label{cost}
    \timecost(\ess{t}) := \vol(S_{0}) + \sum_{j=1}^{t}
    \Big( \vol(S_{j} \Delta S_{j-1})  + \partial(S_{j-1}) \Big),
  \end{align}
where $\Delta$ denotes the symmetric difference between two sets.
\end{definition}
\begin{definition}\label{def:stopping}
  Given integers $T$ and $B$, let $\tau(T,B)$ be the first time one of the following occurs: 
  \begin{enumerate}
    \item $\phi(S_{t}) < \theta_{T}$.
    \item $t = T$ or $\timecost(\ess{t}) > B$.
  \end{enumerate}
\end{definition}

The following theorem shows that $\fastsample$
generates a sample path from the volume-biased ESP with stopping rule
$\tau(T,B)$, and that its complexity is at most $O(\log n)$ times the cost of the sample path it generates.  The complexity is also bounded by $O(B \log n)$.
\begin{theorem} \label{thm:fastsample}
  The algorithm $\fastsample(x,T,B)$ takes as input a vertex $x$,
  an integer $T \geq 0$ and an integer $B \geq 0$.  
  Let $S_{0} = \{x\}$ and let $\tau = \tau(T,B)$.
  The algorithm generates a sample path $(\ess{\tau})$ and
  outputs the last set $S_{\tau}$.  The following hold.
  \begin{enumerate}
    \item 
      The probability that $\fastsample$ generates the sample path $(\ess{\tau})$ is $\Ph_{x} [ \ess{\tau}]$.
    \item If $\fastsample$ generates $(\ess{\tau})$, then its output is 
      $S_{\tau}$ and its complexity is 
      \[ O(\log n) \min(B,\timecost(\ess{\tau})) . \]
  \end{enumerate}
\end{theorem}
The description of $\fastsample$ and the proof of Theorem~\ref{thm:fastsample} are given in Section~\ref{sec:fastsample}.
At a high level, $\fastsample$ simulates the volume-biased ESP 
by updating the boundary of the current set at each step.
We define $\delta(S)$ to be the {\bf two-sided vertex boundary} of $S$,
\[
\delta(S) = 
\{ y : y \in S \wedge e(y,S^{c})>0 \}
\cup
\{ y : y \in S^{c} \wedge e(y,S)>0 \}
.\]
The algorithm maintains a dynamic data structure that stores the current state $S$, 
its two-sided boundary $\delta(S)$, and the values of $p(y,S)$ for vertices in the two-sided boundary.
This allows the algorithm to ignore the vertices in the interior of the set when selecting the next state.
The complexity of $\fastsample$ is dominated by the work required to iterate over the boundary of the current set, select the next state using the stored values of $p(y,S)$, and update the set-with-boundary data structure.

In the following theorem, we bound the expected ratio between the cost of the sample path $(\ess{\tau})$ and the volume of $S_{\tau}$,
which bounds the work/volume ratio for $\evocut$.
\begin{theorem}\label{thm:expectedcost}
  For any starting set $S_{0}$ and any stopping time $\tau$ that is bounded by $T$, we have 
  \begin{align*}
    \Eh_{S_{0}}\left[ \frac{\timecost(\path{0}{\tau})}{\vol(S_{\tau})} \right] \leq 
    1+4\sqrt{T \log \mu(V)}.
  \end{align*}
\end{theorem}
The proof of Theorem~\ref{thm:expectedcost} is the technical highlight of our analysis.
The proof uses a martingale argument and the transform between the ESP and volume-biased ESP.
It bounds the work/volume ratio in a more direct way than previous local partitioning algorithms,
which required the user to guess the approximate volume of the output set~\cite{Spielman:2004,Spielman:2008,Andersen:2006}. 
\begin{proof}
Let $c_{j}$ be the cost of the step in which $S_{j}$ is selected,
\[c_{j} := \vol(S_{j} \Delta S_{j-1}) + \partial(S_{j-1}).\]
We define $c_{0}: = \vol(S_{0})$, and recall 
that $\timecost(\ess{t}) = c_{0} + \ldots + c_{t}$.

Consider the conditional expectation of $c_{j}$ in the ESP.
We have
\begin{align*}
\E\Big( c_{j} \, \Big| \, S_{j-1} \Big) 
&= \E\Big( \vol(S_{j} \Delta S_{j-1}) \, \Big| \, S_{j-1} \Big) + \partial(S_{j-1}).
\end{align*}
We now compute the expected volume of the symmetric difference.
Let $U$ be the uniform random threshold used to select $S_{j}$ from $S_{j-1}$ in the ESP, and recall that $S_{j} \subseteq S_{j-1}$ when $U \geq 1/2$, and $S_{j-1} \subseteq S_{j}$ when $U < 1/2$.  
By Proposition~\ref{prop:conditional},
\begin{align*}
\E\Big( \vol(S_{j} \Delta S_{j-1}) \, \Big| \, S_{j-1} \Big)
&= \E\Big( |\vol(S_{j}) - \vol(S_{j-1})| \, \Big| \, S_{j-1} \Big)\\
&=
\frac{1}{2}
\E\Big( \vol(S_{j}) - \vol(S_{j-1}) \, \Big| \, S_{j-1}, U_{j} < \half \Big)\\
&\quad+
\frac{1}{2}
\E\Big( \vol(S_{j}) - \vol(S_{j-1}) \mid S_{j-1}, U_{j} \geq \half \Big)\\
&= \frac{1}{2} \partial(S_{j-1}) + \frac{1}{2} \partial(S_{j-1})\\
&= \partial(S_{j-1}).
\end{align*}
Therefore, 
$\E( c_{j} \mid S_{j-1} ) = 2\partial(S_{j-1})$.

Let $\R{t} := \frac{\C{t}}{\vol(S_{t})}$, and consider the conditional expectation of $R_{t}$ in the volume-biased ESP.
By Proposition~\ref{prop:transform},
\begin{align*}
\Eh ( \R{t} \mid \ess{t-1} )
&=
\E \left( \frac{\C{t}}{\setvol{t}} \frac{\setvol{t}}{\setvol{t-1}} \, \Big| \, \ess{t-1} \right)\\
&=
\frac{1}{\setvol{t-1}} \Big(\C{t-1} + \E ( c_{t} \mid \ess{t-1} ) \Big)\\
&=
\frac{1}{\setvol{t-1}}\left(\C{t-1} + 2\partial(S_{t-1}) \right)\\
&=
\R{t-1} + 2\phi_{t-1}.
\end{align*}
	
We define 
\[
M_{t} := R_{t} - \sumphi{t}, \qquad \text{where} \qquad  
\sumphi{t} := 1 + 2\sum_{j=1}^{t} \phi_{j-1}.
\]
By construction, 
$(M_{t})$ is a martingale in the volume-biased ESP. 
Notice that $M_{0} = R_{0}-1 = 0$.
Now let $\tau$ be an arbitrary stopping time that is bounded by $T$.
By the optional stopping theorem for martingales (see~\cite{Williams}), we have $\Eh[M_{\tau}] = M_{0} = 0$,
so $\Eh[R_{\tau}] = \Eh[\sumphi{\tau}]$.
By Cauchy-Schwarz,
\[
\sum_{j=0}^{T-1}\phi_{j} 
\leq 
\sqrt{T} \sqrt{\sum_{j=0}^{T-1}{\phi^{2}_{j}}}.
\]
By Jensen's inequality,
\begin{align*}
\Eh[R_{\tau}] =
\Eh[\sumphi{\tau}]
\leq
\Eh[\sumphi{T}]
&=
1 + 2\sum_{j=0}^{T-1} \phi_{j}\\
&\leq
1 + 2\sqrt{T} \Eh \left[ \sqrt{\sum_{j=0}^{T-1}}{\phi^{2}_{j}} \right]\\
&\leq
1 + 2\sqrt{T}  \sqrt{\Eh \left[\sum_{j=0}^{T-1}{\phi^{2}_{j}} \right]}\\
&\leq
1 + 4\sqrt{T\log \mu(V)}
.
\end{align*}
In the last step, we used Lemma~\ref{L:expectation}.
\end{proof}

We can now state the local partitioning algorithm $\evocut$ and prove the main theorem.  

\noindent
\framebox{
\begin{minipage}{\textwidth}
  {\noindent \tt \evocut}$(v,\phi)$:

  \begin{enumerate}
    \item Let $T = \lfloor \phi^{-1}/100 \rfloor$.
      If $T = 0$, then output $\{ v \}$.
    \item Let $S = \fastsample(v,T,\infty)$, and output $S$.
  \end{enumerate}

\end{minipage}
}

\begin{proof}[Proof of Theorem~\ref{thm:evocut}]
  If $T=0$, then we output the set $\{v\}$,
  which trivially satisfies the theorem.  
  Assume that $T \geq 1$. Run $\fastsample(v,T,\infty)$,
  let $w$ be the work (the complexity of the algorithm on this particular run), 
  and let $(\ess{\tau})$ be the sample path generated.
  Since $\phi(\targetset) \leq \phi \leq 1/100T$, 
  Theorem~\ref{thm:partition} shows that 
  with probability at least $7/9$, $S_{\tau}$ satisfies
  $\vol(S_{\tau}) \leq (3/4) \vol(V)$
  and $\phi(S_{\tau}) = O(\sqrt{T^{-1}\log m}) = O(\sqrt{\phi \log n}).$ 
  
  The work $w$ is dominated by the complexity of $\fastsample$, which 
  by Theorem~\ref{thm:fastsample} is $ O(\log n) \cdot \timecost(\ess{\tau})$.
  By Theorem~\ref{thm:expectedcost},
  \[
    \E\left[\frac{\timecost(\ess{\tau})}{\vol(S_{\tau})}\right]
    = O(\sqrt{T \log n}) = O(\sqrt{\phi^{-1} \log n}),
  \]
   so $\E[w/\vol(S_{\tau})] = O(\sqrt{ \phi^{-1} \log n}) \cdot O(\log n)$.
\end{proof}
We remark that the improvement in running time of $\evocut$ comes from ignoring the interior of 
the current set when simulating the volume-biased ESP.  This type of optimization seems unique to evolving sets, and is not possible
with random walks or personalized PageRank.  
We remark that, if we simulated the volume-biased ESP with a naive method that requires 
work roughly proportional to the sum of the volumes of the sets, 
the resulting work/volume ratio would be $O(T \log m)$ rather than $O(\sqrt{T \log m})$. This would match the previous
fastest local partitioning algorithm from \cite{Andersen:2006}.

\subsection{$\fastsample$}\label{sec:fastsample}
In this section we describe $\fastsample$ and prove Theorem~\ref{thm:fastsample}.
The following proposition describes the set-with-boundary data structure that is used by $\fastsample$.  
\begin{proposition}\label{prop:setwithboundary}
  There is a set-with-a-boundary data structure $S$ that supports these operations:
  \begin{itemize}
    \item {\bf add} or {\bf remove} a vertex $y$ from $S$ in time $O(d(y) \log \vol(S))$. 
    \item {\bf get} the value of $e(y,S)$, $\ind(y \in S)$, or $p(y,S)$ in time $O(\log \vol(S))$.
    \item {\bf iterate} over the vertices in the boundary $\delta(S)$ in time $O(|\delta(S)|)$.
  \end{itemize}
\end{proposition}
\begin{proof}[Proof]
The set-with-boundary data structure can be implemented using two standard dictionary data structures. 
We maintain a {\em membership dictionary} $\mathcal{M}$ that contains the vertices in the set $S$, and a {\em boundary dictionary} $\mathcal{B}$ that contains the vertices in $\delta(S)$ and stores the associated value $\mathcal{B}(y) = e(y,S)$ for each $y \in \delta(S)$.  These dictionaries must support the following operations: inserting and deleting a key and its value, checking whether a given key is in the dictionary, and looking up the value associated with a key.  A red/black tree supports these operations in $O(\log N)$ worst case time per operation, where $N$ is the number of keys currently in the tree (see~\cite{CLRS}).

The value of $\ind(y \in S)$ can be computed by checking whether $y \in \mathcal{M}$.
For any vertex $y \in V$, the value of $e(y,S)$ can be computed
using two lookups, one into the membership dictionary and one into the boundary dictionary:
\begin{align}\label{edgecases}
e(y,S) = 
\begin{cases}
  \mathcal{B}(y) & \text{if $y \in \mathcal{B}$,} \\ 
  d(y) & \text{if $y \not \in \mathcal{B}$ and $y \in \mathcal{M}$,} \\
  0    & \text{if $y \not \in \mathcal{B}$ and $y  \not \in \mathcal{M}$.}
\end{cases}
\end{align}
It is straightforward to compute $p(y,S)$ from $e(y,S)$ and $\ind(y \in S)$.

Each time a node is added or removed from $S$, we update the membership
dictionary. For each neighbor $z \sim y$ we
increment or decrement the value of $e(z,S)$ in the boundary dictionary. For each node $z$ that was updated (including $y$ and its neighbors), we determine whether the node is contained in $\delta(S)$ by examining the values of $e(z,S)$ and $(z \in S)$, then add or remove $z$ from the boundary dictionary when necessary. In total, 
adding or removing $y$ takes $O(d(y))$ dictionary operations. 
The size of either dictionary is $O(\vol(S))$, so each dictionary operation takes time $O(\log \vol(S))$.  
\end{proof}

We now describe $\fastsample$.
The input is a starting vertex $x$, a time limit $T\geq0$, and a budget $B\geq0$.  The output is a set $S_{\tau}$ sampled from the volume-biased ESP with the stopping rule $\tau = \tau(T,B)$.
The algorithm simulates the volume-biased ESP using the coupling described in Proposition~\ref{prop:coupling}.  It uses an instance $S$ of the set-with-boundary data structure to maintain the current state $S_{t}$, and also stores a vertex $X$ that represents the current walk position $X_{t}$.  
Initially, $S = S_{0} = \{x\}$ and $X = X_{0} = x$.
The algorithm proceeds in steps. At the beginning of step $t$, we have $S = S_{t-1}$ and $X = X_{t-1}$.
The algorithm continues until the stopping time $\tau$
is reached, then outputs $S_{\tau}$.   

Each step has two stages.  
In the first stage we select $S_{t}$ and compute a list of the vertices that need to be added or removed from $S_{t-1}$ to form $S_{t}$. 
This stage requires $O(1)+O(\partial(S_{t-1}))$ operations.
We stop after the first stage if $\timecost(\path{0}{t}) >B$. 
Otherwise, we proceed to the second stage in which we update $S$ to $S_{t}$, which requires $O(1)+O(\vol(S_{t} \Delta S_{t-1}))$ operations.  
Each operation is either a constant time operation or a dictionary operation requiring time $O(\log n)$.  

In stage 1, we begin with $X=X_{t-1}$, then update the walk particle.
Given that $X_{t-1} = x_{t-1}$, we choose $X_{t} = x_{t}$ with probability $p(x_{t-1},x_{t})$, and update $X=X_{t}$. 
We assume that a random neighbor of $x_{t-1}$ can be selected in constant time.
We compute $p(x_{t},S)$ by a lookup into the set-with-boundary data structure, and select a random threshold $Z$ uniformly from the interval $[0,p(x_{t},S)]$. At this point we define $S_{t} = \{y \mid p(y,S_{t-1}) \geq Z\}$, but we do not yet update $S$ to reflect $S_{t}$.  Instead 
we create a list $D$ of the vertices in the set difference $S_{t} \Delta S_{t-1}$.
We populate the list by iterating over each node $y \in \delta(S)$, looking up the value of $p(y,S)$, and comparing this value with the threshold $Z$.
While doing this, we update the values of $\vol(S_{t})$ and $\timecost(\ess{t})$.  We then check whether either of the stopping conditions $t=T$ or $\timecost(\ess{t}) > B$ is satisfied.  
If so, then $\tau = t$, so we stop and output $S_{t} = S_{t-1} \Delta D$.  Otherwise, we proceed to the next stage.

In stage 2, we update $S$ to $S_{t}$ by adding or removing the vertices from $D$.  While making these updates to $S$, we also update 
$\partial(S_{t-1})$ to $\partial(S_{t})$.
We compute $\phi(S_{t}) = \partial(S_{t})/\vol(S_{t})$ and check whether $\phi(S_{t}) < \theta_{T}$.
If so, then we halt and output the set $S_{t}$.  Otherwise, we proceed to the next step.

\noindent
\framebox{
\begin{minipage}{\textwidth}
	\begin{small}
    {\noindent \tt \fastsample}$(x,T,B)$:

    \vspace{.1cm} {\bf Input and output:}\\
    The input is a starting vertex $x$, and two integers $T\geq0$ and $B\geq0$.\\
    The output is a set $S_{\tau}$ sampled from the volume-biased ESP 
    with stopping rule $\tau = \tau(T,B)$.  

    \vspace{.1cm} {\bf Internal state:}\\
    $S$ = an instance of the set-with-boundary data structure.\\
    $X$ = the current location of the random walk particle.\\
    The algorithm also maintains the current values of 
    $\partial(S)$, $\vol(S)$, and $\timecost(\ess{})$.
    
    \vspace{.1cm} {\bf Algorithm:}\\
    Initially, let $S = S_{0} = \{x\}$ and $X = x_{0} = x$.\\
    At the beginning of step $t$, we have $S = S_{t-1}$ and $X = X_{t-1}$. \\
    For $t=1 \ldots \tau$, do step $t$ as follows:
    \begin{enumerate}
      \item {\bf Stage 1.}  Select the vertices to add or remove from $S$:
		    \begin{enumerate}
              \item   Given $X_{t-1}=x_{t-1}$, select $X_{t}=x_{t}$ with probability $p(x_{t-1},x_{t})$ and update $X \leftarrow x_{t}$.
                \item Lookup $p(x_{t},S_{t-1})$ and pick $Z$ uniformly at random from the interval $[0,p(x_{t},S_{t-1})]$. 
                \item Define $S_{t} = \{y \mid p(y,S_{t-1}) \geq Z\}$.
                \item Compute a list $D$ of the vertices in $S_{t} \Delta S_{t-1}$ as follows. \\
                  For each $y \in \delta(S_{t-1})$:
		    \begin{enumerate}
              \item lookup $p(y,S_{t-1})$.
              \item if $y \in S_{t} \Delta S_{t-1}$, then add $y$ to $D$.
		    \end{enumerate}
                \item Update $\vol(S_{t})$ and $\timecost(\ess{t})$.
                \item If $t=T$ or $\timecost(\ess{t}) > B$, then $t = \tau$, 
                  so halt and output $S_{t} = S_{t-1} \Delta D$.\\
                  Otherwise, proceed to the next stage.
	\end{enumerate}

\item {\bf Stage 2.} Update $S$:
\begin{enumerate}
\item Update $S$ to $S_{t} = S_{t-1} \Delta D$ by adding or removing the vertices in $D$ from $S$.
\item Update $\partial(S_{t})$ and compute $\phi(S_{t}) = \partial(S_{t})/\vol(S_{t})$. 
\item If $\phi(S_{t}) < \theta_{T}$, then $t = \tau$, so halt and output $S_{t}$.\\
	  Otherwise, proceed to the next step.
\end{enumerate}

\end{enumerate}

\end{small}
\end{minipage}
}
      
\begin{proof}[Proof of Theorem~\ref{thm:fastsample}]
By construction, $\fastsample$ simulates the coupling
from Section~\ref{sec:coupling}.  
By Proposition~\ref{prop:coupling}, the sequence $(\path{0}{\tau})$ it generates is a sample path from the volume-biased ESP. 

Let $c_{t} = O(\partial(S_{t-1}) + \vol(S_{t-1} \Delta S_{t}))$.
We will show that the number of operations performed in step $t$
is $O(c_{t})$.  Each operation is either a constant time operation or a dictionary operation requiring time $O(\log n)$. 
The number of operations performed in stage 1 is
dominated by part (d) in the pseudocode, in which lookup operations are performed for each vertex in $O(\delta(S_{t-1}))$.  This requires $O(1)+O(\partial(S_{t}))$ operations.
The number of operations performed in stage 2 is dominated by part (a), in which the vertices from $S_{t-1} \Delta S_{t}$ are added or removed from the set-with-boundary data structure $S$. 
By Proposition~\ref{prop:setwithboundary}, this requires $O(1) + O(\vol(\setdiff{t}{t-1}))$ operations.
In total, the number of operations required in step $t$ is 
\begin{align*}
O(\partial(S_{t-1}) + \vol(S_{t-1} \Delta S_{t})) + O(1)
= O(c_{t}).
\end{align*}
The $O(1)$ term above can be ignored safely because $\partial(S_{t-1}) > 0$ for $t \leq \tau$.
We define $c_{0} = d(x_{0}) = \mu(S_{0})$ to account for the work required to create $S_{0}$.
Then, the total number of operations performed by $\fastsample$
is $O(c_{0} +, \ldots, + c_{\tau}) = O(\timecost(\ess{\tau}))$.

If $\timecost(\path{0}{\tau}) > B$, then $\fastsample$ halts after stage 1 during step $\tau$, and the number of operations performed in step $\tau$ is $O(\partial(S_{\tau-1}))$.
The total number of operations performed is therefore 
\[O(\timecost(\ess{\tau-1})) + \partial(S_{\tau-1}) = O(2B),\]
because
$\partial(S_{\tau-1}) \leq \vol(S_{\tau-1}) \leq \timecost(\ess{\tau-1}) \leq B$.  The theorem follows.
\end{proof}

\section{Finding balanced cuts}\label{sec:balanced}

Spielman and Teng used their local partitioning algorithm $\justnibble$ to construct an algorithm $\partition$ that finds a cut with small conductance and approximately maximal volume~\cite{Spielman:2008}.
To do this, they constructed a subroutine called $\nibble$
that applies their local partitioning algorithm from a random starting vertex with a random budget.  The time complexity of $\nibble$ is nearly independent of the size of the input graph, and the set that it output contains, in expectation, a small fraction of any set in the graph that has sufficiently small conductance.  In this section, we construct an analogous subroutine called $\evonibble$ by making small modifications to $\evocut$. We then describe the algorithm $\evopartition$ that 
results from substituting $\evonibble$ for $\nibble$ 
in Spielman and Teng's construction of $\partition$.

\subsection{$\evonibble$}
In this section we describe the subroutine $\evonibble$.

\noindent
\framebox{
\begin{minipage}{\textwidth}
    {\noindent \tt \evonibble}$(\phi)$:
\begin{enumerate}
  \item Let $T = \lfloor \phi^{-1}/100 \rfloor$, and let $\theta_{T} = \sqrt{4T^{-1}\log \mu(V)}$.
  \item Choose a random vertex $X\in V$ with probability 
      $P(X=x) = d(x)/\vol(V)$. 
  \item Choose a random budget as follows. 
    Let $\Jmax = \lceil \log_{2} \vol(V) \rceil$, and let $J$ be an integer from $[0,\Jmax]$ chosen 
    with probability $P(J=j) = \sigma 2^{-j}$, where $\sigma$ is a proportionality constant.  Let $B_{J} = 8\gamma2^{J}$, where $\gamma = 1+4\sqrt{T \log \vol(V)}$.
  \item Compute $S = \fastsample(X,T,B_{J})$.
  \item
    If $\phi(S) \leq 3\theta_{T}$ and $\vol(S) \leq (3/4) \vol(V)$, then output $S$. 
    Otherwise output $\emptyset$.
\end{enumerate}
\end{minipage}
}

\begin{theorem}\label{thm:evonibble}
  The randomized algorithm $\evonibble(\phi)$ takes as input $\phi \in (0,1)$ and outputs a set $S \subseteq V$.
  The following hold:
  \begin{enumerate}
    \item The expected complexity is $O( \phi^{-1/2}\log^{5/2} \vol(V))$.
    \item Either $S = \emptyset$, or $S$ satisfies
      $\phi(S) = O(\sqrt{\phi \log \vol(V)})$ and $\vol(S) \leq (3/4) \vol(V)$.
    \item For any set $\targetset \subseteq V$ that satisfies $\vol(\targetset) \leq (2/3) \vol(V)$ and \mbox{$\phi(\targetset) \leq \phi$},
      \[ \E \left( \frac{\vol(S \cap \targetset)}{\vol(\targetset)} \right)  \geq \frac{1}{20\vol(V)}.\]
  \end{enumerate}
\end{theorem}
\begin{proof}
   First we prove conclusion (1).  Let $W$ be the complexity of the algorithm.
   By Theorem~\ref{thm:fastsample}, 
   we have $(W\mid J) = O(B_{J} \log \vol(V)) = O(\gamma 2^{J}\log \vol(V))$, where $\gamma = O(\sqrt{\phi^{-1} \log \vol(V)})$.
   The expected complexity is 
   \begin{align*}
       E(W)
       &= \sum_{j \in [0,\Jmax]} E(W | J = j) P(J=j)\\
       &= \sum_{j \in [0,\Jmax]} O(\gamma 2^{j}\log \vol(V)) O(2^{-j})\\
       &= O(\gamma \log \vol(V) \Jmax)\\
       &= O( \phi^{-1/2}\log \vol(V)^{5/2}).
   \end{align*}

   Conclusion (2) is immediate from the definition of the algorithm.  We now prove conclusion (3).
   Let $S_{out}$ be the output of $\evonibble$.  Let $X$ be the starting vertex.
   Let $A$ be a set that satisfies the requirements of conclusion (3), 
   and let $A_{T}\subseteq A$ be the subset described in Proposition~\ref{prop:exitprob}.
   We will prove the following:
   \begin{equation}\label{nibbleeqn} 
     \text{if $x \in A_{T}$, then $E(\mu(S_{out} \cap A) \mid X=x) \geq 1/10$.}
   \end{equation}
   After that, conclusion (3) follows by taking the expectation over the choice of the starting vertex:
   \begin{align*}
     E(\mu(S_{out} \cap A))
       =    \sum_{x \in V} E(\mu(S_{out} \cap A \mid X=x) P(X=x)
       \geq (1/10) P(X \in A_{T})
       \geq (1/20) \vol(A)/\vol(V).
   \end{align*}
   We now prove~\eqref{nibbleeqn}.
   Consider a sample path from the volume-biased ESP started from $\{X\}$,
   and let $\tau = \tau(T,\infty)$.
   Let $D$ be the event that all of the following hold:
   \begin{enumerate}
     \item $\timecost(\path{0}{\tau}) \leq 4\gamma \vol(S_{\tau})$,
     \item $\phi(S_{\tau}) \leq 3\theta_{T}$,
     \item $\vol(S_{\tau}) \leq (3/4) \vol(V)$,
     \item $\vol(S_{\tau} \cap \targetset) \geq (9/10) \vol(S_{\tau})$. 
   \end{enumerate}
   Combining Theorem~\ref{thm:expectedcost} and Theorem~\ref{thm:partition} shows that if $x \in A_{T}$, 
   then \mbox{$P(D \mid X = x) \geq 1/4$}. 
   The subroutine $\fastsample(X,T,B)$ returns the set $S_{\tau(T,B)}$ rather than $S_{\tau}$.  To deal with this,
   we define the events \mbox{$D_{j} = D \wedge (\vol(S_{\tau}) \in [2^{j},2^{j+1}))$.}
   Note that
   \[(\timecost(\path{0}{\tau}) \mid D_{j}) \leq 8\gamma 2^{j} = B_{j}.\]
   This implies that if the event $(D_{j} \wedge (J\geq j))$ holds,
   then $S_{out} = S_{\tau(T,B)} = S_{\tau}$, 
   and furthermore \mbox{$\mu(S_{out} \cap A) \geq (9/10) 2^{j}$.}
   For any $x \in A_{T}$, we have 
   \begin{align*}
             E(\mu(S_{out} \cap A) \mid X=x) 
     &\geq   \sum_{j \in [0,\Jmax]} E(\mu(S_{out} \cap A) \mid X=x, J=j, D_{j}) P(D_{j} \mid X=x) P(J=j) \\
     &\geq   \sum_{j \in [0,\Jmax]} (9/10)2^{j} (\sigma 2^{-j}) P(D_{j} \mid X=x) \\
     &\geq   (9/10) \sigma \sum_{j \in [0,\Jmax]} P(D_{j} \mid X=x) \\
     &\geq   (9/10) \sigma P(D \mid X=x)\\
     &\geq   1/10.
   \end{align*}
   This establishes~\eqref{nibbleeqn} and completes the proof.
\end{proof}

\subsection{$\evopartition$}
The algorithm $\evopartition$ described in the following theorem
can be constructed by substituting the subroutine $\evonibble$ for $\nibble$
in Spielman and Teng's algorithm $\partition$.
We omit the proof of Theorem~\ref{thm:balanced} and the description of the algorithm, and refer the reader to Theorem 3.2 in~\cite{Spielman:2008}.  At a high level, the algorithm applies the nibbling subroutine, removes the resulting cut from the graph, then repeats.  It stops after $\polylognoparen{m}$ steps or when a large fraction of the graph has been removed.
\begin{theorem}\label{thm:balanced}
  The randomized algorithm $\evopartition(\phi)$ takes an input $\phi \in (0,1)$, and it outputs a set $S \subseteq V$.
  The expected complexity is $\polylognoparen{m\phi^{-1/2}}$.
     With probability at least $1/2$, both of the following hold: 
  \begin{enumerate}
   \item $\phi(S) = O(\sqrt{\phi \log m})$ and $\vol(S) \leq (7/8) \vol(V)$.
   \item At least one of the following holds:
      \begin{enumerate}
       \item $\vol(S) \geq (1/4) \vol(V)$
       \item For any set $\targetset \subseteq V$ that satisfies $\phi(\targetset) \leq \phi$ and $\vol(\targetset) \leq (2/3) \vol(V)$, 
         we have $\vol(S \cap \targetset) \geq \vol(\targetset)/2$. 
      \end{enumerate}
 \end{enumerate}
\end{theorem}
We remark that the complexity of $\evopartition$ can be reduced from $\polylognoparen{m\phi^{-1/2}}$ to $\polylog{m+n\phi^{-1/2}}$  
by applying the sparsification technique of Bencz{\'u}r-Karger~\cite{BenczurKarger}. 

In Table~\ref{fig:balanced}, we summarize the complexity and approximation guarantee
of selected algorithms for the balanced cut problem.  
For all the algorithms listed, we first apply the Bencz{\'u}r-Karger~\cite{BenczurKarger} sparsification technique to the graph.
We state the running times in terms of $\phi$, which for the first four algorithms
is specified as part of the input (see Theorem~\ref{thm:balanced}).  The next three algorithms solve a different formulation of the balanced cut problem,
where the volume of the set $A \subseteq V$ is specified rather than the conductance.  For the purpose of comparison, we 
translate their approximation guarantees to our formulation.  See the original papers for the precise statements.

\begin{table}\begin{center}
\begin{tabular}{|l|ll|}
    \hline
  balanced cut algorithm &  complexity &  approximation\\
    \hline
    $\partition$ (ST04)~\cite{Spielman:2004}        & $\polylog{m + n\phi^{-5/3}}$  & $\phi \rightarrow O(\phi^{1/3}\log^{2/3}n)$\\
    $\partition$ (ST08)~\cite{Spielman:2008}        & $\polylog{m + n\phi^{-2}}$    & $\phi \rightarrow O(\phi^{1/2}\log^{3/2}n)$\\
    $\prpartition$ (ACL06)~\cite{Andersen:2006}    & $\polylog{m + n\phi^{-1}}$      & $\phi \rightarrow O(\phi^{1/2}\log^{1/2}n)$\\
    $\evopartition$ (this paper) & $\polylog{m + n\phi^{-1/2}}$              & $\phi \rightarrow O(\phi^{1/2}\log^{1/2}n)$\\
    \hline
    Arora-Hazan-Kale~\cite{AroraHazanKale} & $\polylognoparen{n^{2}}$                            & $\phi \rightarrow O(\phi \log^{1/2} n)$ \\
    Arora-Kale~\cite{Arora:2007}           & $\polylog{m + \min\{n^{3/2},n\phi^{-1}\}}$   & $\phi \rightarrow O(\phi \log n)$\\
    Orecchia et al.~\cite{Orecchia:2008}   & $\polylog{m + \min\{n^{3/2},n\phi^{-1}\}}$   & $\phi \rightarrow O(\phi \log n)$\\
    \hline
    Recursive spectral (power method)   & $\polylognoparen{n^{2}\lambda^{-1}}$        & $\phi \rightarrow O(\phi^{1/2})$\\
    Recursive spectral (Lanczos) & $\polylognoparen{n^{2}\lambda^{-1/2}}$   & \mbox{$\phi \rightarrow O(\phi^{1/2})$}\\
    Recursive spectral (ST solver)~\cite{DBLP:journals/corr/abs-cs-0607105}& $\polylognoparen{n^{2}}$   & \mbox{$\phi \rightarrow O(\phi^{1/2})$}\\
    \hline
\end{tabular}\hfill
\end{center}
\caption{Comparison of selected algorithms for finding balanced cuts.
         Here $n = |V|$, $m = |E|$, and $\phi$ is an input to the algorithm. 
         }\label{fig:balanced}
\end{table}

Spielman and Teng's algorithm $\partition$~\cite{Spielman:2004,Spielman:2008}
was the first balanced cut algorithm with a nearly-linear complexity.
The algorithm of Khandekar-Rao-Vazirani~\cite{Khandekar:2006} finds a balanced cut by solving single commodity
flow problems, and outputs a cut of
conductance $O(\phi \log^{2} n)$ in time $\polylognoparen{T_{flow}}$, where
$T_{flow} = \polylog{m + \min( n^{3/2},n\phi^{-1})}$.
The algorithm of Arora and Kale~\cite{Arora:2007} 
outputs a cut of conductance $O(\phi \log n)$ in time
$\polylognoparen{T_{flow}}$.  
The algorithm of Orecchia et al.~\cite{Orecchia:2008} obtains the same running time and approximation 
within the cut-matching framework of~\cite{Khandekar:2006}.
The best approximation is $O(\phi \log^{1/2} n)$, due to Arora-Rao-Vazirani~\cite{ARV}.
The fastest algorithm that attains this approximation is due to Arora-Hazan-Kale~\cite{AroraHazanKale}.

Spectral partitioning methods can be applied recursively to find balanced cuts (see~\cite{Spielman:1996}).  
As far as we know, there is no good way to lower bound the balance of the cut output by the basic spectral partitioning method. As a result,
the best known bound on the depth of the recursion is $\Omega(n)$, and
the best known time bound for finding balanced cuts using recursive spectral partitioning is $\Omega(n^{2})$. 
Let $\phi_{*} = \min\{ \phi(A) : \, A \subseteq V, \, \vol(A) \leq \vol(V)/2\}$. Let $\lambda$ be the smallest nonzero eigenvalue of the normalized Laplacian, which satisfies $\lambda \leq 2\phi_{*}$.
The basic spectral partitioning method produces a set of conductance $O(\sqrt{\lambda}) = O(\sqrt{\phi_{*}})$~(see \cite{Mihail:1989, Chung:1997}).
If the power method is used to compute an approximate eigenvector, then the complexity of finding an unbalanced cut using the spectral method is $\polylognoparen{n\lambda^{-1}}$.  If the Lanczos algorithm is used instead, then the complexity improves to $\polylognoparen{n\lambda^{-1/2}}$.  The complexity can be further improved to $\polylognoparen{n}$
by using the linear system solver of Spielman-Teng to compute the pseudo-inverse, then applying the inverse power method (see~\cite{DBLP:journals/corr/abs-cs-0607105}). 

\bibliographystyle{abbrv}
\bibliography{evolving}

\end{document}